\documentclass[journal]{IEEEtran}

\usepackage[latin1]{inputenc}
\usepackage{amssymb, amsmath, amsfonts, amsthm}
\usepackage{algorithm, algorithmic, dsfont, caption}
\usepackage{float, graphics, xspace, graphicx}
\usepackage[usenames,dvipsnames]{color}
\usepackage{cite, epsfig, tikz, standalone, subcaption}

\newcommand{\PP}{\mathbb{P}}
\newcommand{\EE}{\mathbb{E}}
\newcommand{\dd}{\mathrm{d}}
\newcommand{\al}{\alpha}
\newcommand{\be}{\beta}
\newcommand{\s}{\sigma^2}
\newcommand{\la}{\lambda}
\newcommand{\intd}{\int_\zeta^\xi}

\newtheorem{theorem}{Theorem}
\newtheorem{lemma}{Lemma}
\newtheorem{proposition}{Proposition}

\IEEEoverridecommandlockouts
\allowdisplaybreaks

\begin{document}

\title{Backscatter Communications for Wireless Powered Sensor Networks with Collision Resolution\vspace{-1mm}}

\author{Constantinos Psomas, \IEEEmembership{Member, IEEE}, and Ioannis Krikidis, \IEEEmembership{Senior Member, IEEE}\vspace{-6mm}
\thanks{C. Psomas and I. Krikidis are with the KIOS Research and Innovation Center of Excellence, University of Cyprus, Cyprus (e-mail: \{psomas, krikidis\}@ucy.ac.cy). 

This work was supported by the Research Promotion Foundation, Cyprus under the project COM-MED with pr. no. KOINA/ERANETMED/1114/03.
}}

\maketitle

\begin{abstract}
Wireless powered backscatter communications is an attractive technology for next-generation low-powered sensor networks such as the Internet of Things. However, backscattering suffers from collisions due to multiple simultaneous transmissions and a dyadic backscatter channel, which greatly attenuate the received signal at the reader. This letter deals with backscatter communications in sensor networks from a large-scale point-of-view and considers various collision resolution techniques: directional antennas, ultra-narrow band transmissions and successive interference cancellation. We derive analytical expressions for the decoding probability and our results show the significant gains, which can be achieved from the aforementioned techniques.\vspace{-2mm}
\end{abstract}

\begin{keywords}
Backscatter communications, sensor networks, ultra-narrow band, Poisson point process.\vspace{-3.2mm}
\end{keywords}

\section{Introduction}
The Internet of Things aims at the massive and ubiquitous deployment of sensors. This has raised questions about the energy requirements/consumption of these devices as it will be impractical or impossible, to individually recharge them on a regular basis. A practical method to resolve this issue, is wireless power transfer, i.e. energy harvesting from electromagnetic radiation. Specifically, a wireless powered communication network (WPCN) is a promising architecture for next-generation wireless sensor networks, where terminals power their uplink transmissions by harvesting energy from radio frequency (RF) signals, transmitted by a power beacon or an access point over a dedicated orthogonal channel \cite{ZENG}.

The research community mostly focuses on devices which harvest energy through a rectifying-antenna, a diode-based circuit that converts RF signals to direct-current voltage \cite{KRI}. Even though this approach is feasible, the amount of energy harvested via this operation is usually very small compared to the energy needed to operate the device. Another approach is backscattering, mainly used for RF identification (RFID), where part of the received signal is scattered back to the transceiver using a load modulation scheme \cite{BOOK}. Here, the energy consumption is very low as most of the components used are passive. However, the dyadic channel formed by the forward and backscatter links, greatly affects its performance. In the literature, few studies exist on backscatter communications from a system level standpoint. In \cite{BLE}, the performance of anti-collision techniques is studied in backscatter sensor networks; the authors show that the modulation employed at a sensor is crucial to achieve adequate network performance. The work in \cite{BEK} presents a 3D analytical model of ultra-high frequency RFID systems under cascaded fading channels and studies their performance in terms of the detection probability. Finally, the work in \cite{HUA2} considers backscatter communication networks with multiple power beacons, where each beacon serves a cluster of devices; the performance of the network is studied in terms of coverage probability and achievable rate.

Most works in the literature, consider simple topologies and deterministic spatial models. Also, they do not consider the dyadic backscatter channel, which is critical to the network's performance. Motivated by this, in this letter, we study backscatter sensor networks with spatial randomness to efficiently evaluate the large-scale path-loss effects in backscattering networks. We take into account the dyadic channel and consider the cases where the uplink and downlink channels are partially or fully correlated. Due to the existence of high interference levels (collisions) at the reader from multiple concurrent transmissions, we investigate the combined implementation of three collision resolution techniques to boost the performance. Specifically, we consider directional antennas and successive interference cancellation (SIC) at the reader as well as ultra-narrow band transmissions for multiple access. We derive analytical expressions for the decoding probability, i.e. the probability of decoding a random sensor, using stochastic geometry and show that the combination of these techniques provides massive gains in performance.

\underline{Notation}: $\Re[x]$ and $\Im[x]$ are the real and imaginary parts of $x$, respectively, $\imath = \sqrt{-1}$ is the imaginary unit, $\lVert x \rVert$ is the Euclidean norm of $x$, $\Gamma[\cdot,\cdot]$ is the upper incomplete gamma function, $\mathcal{I}_0(\cdot)$ and $\mathcal{K}_0(\cdot)$ are the modified Bessel functions of the first and second kind of order zero, respectively.\vspace{-1mm}

\section{System model}

\subsubsection{Topology}
Consider a single-cell WPCN with backscatter communications and randomly deployed sensors, utilizing a bandwidth $BW$. The coverage (reading) zone is modeled as a disc of radius $\xi$, with the reader located at the disc's origin with an exclusion zone of radius $\zeta$ around it, $\xi > \zeta \geq 1$. The exclusion zone defines the minimum reading distance between the reader and a sensor, whereas $\xi$ defines the maximum reading distance \cite{BOOK}; the restriction $\zeta \geq 1$ ensures the harvested energy is never greater than the transmitted power. The location of the sensors is modeled as a homogeneous Poisson point process (PPP) $\Psi = \{z_k\}$, $k \geq 1$, with density $\la$ \cite{HAE}; $z_k$ denotes the coordinates of the $k$-th sensor. Each sensor is equipped with a single antenna and a circuit responsible for modulating and backscattering information to the reader.

\subsubsection{Channel model}
We assume all wireless links suffer from small-scale block fading and large-scale path-loss effects. A backscatter channel is dyadic, that is, it is characterized by a forward link (reader to sensor) and a backscatter link (sensor to reader). We denote by $h^i_f$ and $h^i_b$ the channel coefficients for the forward and backscatter link with the $i$-th sensor, respectively; it is clear that the two links can be correlated by a coefficient $\rho \in [0,1]$ \cite{JDG}. The fading of forward and backscatter links is considered to be Rayleigh with variance $\s_f$ and $\s_b$, respectively\footnote{Rayleigh is used for simplicity but other channel models such as Rice or Nakagami could also be considered.}. Therefore, the probability distribution function (pdf) of the backscatter channel power is the pdf of the product of two exponential random variables. The pdf $f_h(h,\rho)$ of the $i$-th sensor's backscatter channel power $h_i = |h^i_f|^2 |h^i_b|^2$, where $h^i_f$ and $h^i_b$ have correlation $\rho$ is
\begin{align}\label{pdf_fad0}
f_h(h,\rho) = \frac{2\mu_f\mu_b}{1-\rho^2} \mathcal{I}_0\left(\frac{2 \rho\sqrt{\mu_f\mu_b h}}{1-\rho^2}\right) \!\mathcal{K}_0\left(\!\frac{2\sqrt{\mu_f\mu_b h}}{1-\rho^2}\right),
\end{align}
for uncorrelated/partial correlated channels, i.e. $\rho \in [0,1)$, and
\begin{align}\label{pdf_fad1}
f_h(h,1) = \frac{1}{2}\sqrt{\frac{\mu_f\mu_b}{h}}\exp\left\{-\sqrt{\mu_f\mu_b h}\right\},
\end{align}
for perfectly correlated channels, i.e. $\rho = 1$, where $\mu_f = 1/\s_f$ and $\mu_b = 1/\s_b$; throughout this work we use $\mu_f = \mu_b = 1$. The above can be derived from the pdfs for the product of two Rayleigh random variables given in \cite{JDG}. The path-loss model assumes that the received power is proportional to $d_i^{-\alpha}$, where $d_i = \lVert z_i \rVert$ is the Euclidean distance from the origin to the $i$-th sensor, $\alpha > 1$ is the path-loss exponent. The pdf of $d_i$ is \cite{JFC}
\begin{align}\label{dist_pdf}
f_d(x) = \frac{1}{\pi(\xi^2-\zeta^2)}.
\end{align}
Finally, we assume that all wireless links exhibit additive white Gaussian noise with variance $\sigma^2$.

\begin{figure}[t]
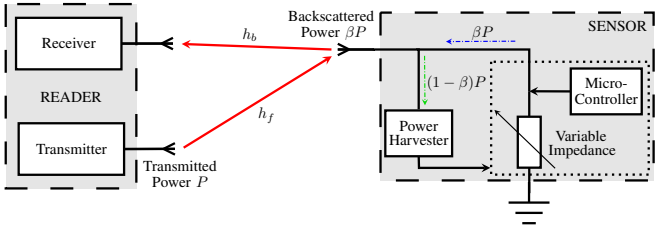
\centering
  \includestandalone[width=\linewidth]{back_comm}
  \caption{Wireless powered backscatter communications.}\label{back_comms}\vspace{-5mm}
\end{figure}

\subsubsection{Backscatter communications}
The reader transmits an unmodulated RF signal $s(t) = \sqrt{2 P} \Re[\exp\{\imath 2\pi ft\}]$ towards the sensors, where $P = \mathbb{E}[s^2(t)]$ is the reader's transmit power and $f$ denotes the carrier frequency. The received signal obtained by a sensor is modulated and reflected back to the reader \cite{BOOK}. We model the portion of the received power reflected back by a sensor with a reflection coefficient $\beta \in [0,1]$; this means $(1-\beta)P$ is used to power the load modulation scheme, which varies the impedance between different states, and $\be P$ is the backscattered power \cite{BOOK, HUA2}. The backscattered power highly depends on the received power, the impedance matching and the energy requirements of the circuit. Fig. \ref{back_comms} schematically depicts the considered backscatter communication scheme. The reader attempts to decode the signal from the $i$-th sensor by the complex baseband equivalent signal $y = \tilde{y}_i + \sum_{j\neq i} \tilde{y}_j+n,$ where $\tilde{y}_i = \sqrt{\be P d_i^{-2 \al}} h^i_f h^i_b x_i$, $x_i$ is the unit-energy transmitted symbol from the $i$-th sensor, and $n$ is a circularly symmetric complex Gaussian random variable with zero mean and variance $\s$. Then, the signal-to-interference-plus-noise ratio (SINR) of the reader for the $i$-th sensor is $\gamma_i = \frac{\be P h_i}{d_i^{2\al}(\s + \be P I_i)},$ where $I_i = \sum_{z_j \in \Psi \backslash \{z_i\}} \frac{h_j}{d_j^{2\al}}$. Note that the path-loss exponent is doubled as a result of the dyadic backscatter channel.\vspace{-1mm}

\section{Collision Resolution Techniques}\label{techniques}
The main issue in backscatter communication networks is the existence of collisions at the reader, due to the simultaneous transmission from multiple sensors, which reduces the reader's ability to decode a sensor's signal. To avoid such collisions, numerous techniques have been proposed \cite{BOOK}. In this work, we investigate the following.

\subsubsection{Sectorized directional antennas}
The reader is equipped with $D$ sectorized directional antennas, i.e. the beamwidth of the antenna is $\frac{2\pi}{D}$. Therefore, each sectorized antenna transmits towards a disk's sector of arc length $\frac{2\pi}{D}$, thus reducing the number of active sensors. This technique is commonly referred to as space division multiple access (SDMA). We assume that the antenna's main lobe gain is given by $G = \frac{D}{1+\epsilon(D-1)}$ where $\epsilon \in [0,1]$ defines the directional efficiency of the antenna \cite{HUN}.

\subsubsection{Ultra-narrow band transmissions}
The sensors employ ultra-narrow band transmissions using a random frequency division multiple access (FDMA) scheme \cite{MO}. Specifically, the $i$-th sensor selects a carrier frequency $f_i$ at random from the available bandwidth $BW$. Let $\delta$ be a system parameter, defining the frequency spacing in which collisions occur. Then, if $|f_i - f_j| < \delta$, $i \neq j$, the signals of the $i$-th and $j$-th sensor collide \cite{MO}. As a result, the collision probability is modeled by $p_c = \delta/BW$.

\subsubsection{Successive interference cancellation}\label{SIC}
We assume the reader employs SIC techniques if the SINR does not achieve the required threshold for a certain sensor \cite{CP}. Specifically, the reader first attempts to decode and remove the strongest interfering signal. If successful, the reader retries to decode the required signal; otherwise, the process is repeated up to $n$ times during which the reader will either manage to decode the required signal or after the $n$-th attempt it will be in outage.

\section{Performance Analysis}
In this section, we derive the probability of the reader decoding the backscattered signal of a random sensor. We first need the following lemma.

\begin{lemma}\label{lem1}
The characteristic function $\phi_I(t,\la,\zeta)$ of the interference term $I$ in a network of density $\la$, where the nearest interfering signal is at minimum distance $\zeta$, is given by
\begin{align}\label{char}
\phi_I(t,\la,\zeta) = \exp\!\left\{\!\pi\la\left(\zeta^2-\xi^2 + \int_0^\infty \Phi(t,h) f_h(h,\rho) \dd h\right)\!\right\},
\end{align}
where
{\small\begin{align}
\Phi(t,h) &= \frac{\left(-\imath t h\right)^\frac{1}{\al}}{\al} \Bigg(\Gamma\left[-\frac{1}{\al}, -\frac{\imath t h}{\xi^{2\al}}\right] - \Gamma\left[-\frac{1}{\al}, -\frac{\imath t h}{\zeta^{2\al}} \right]\Bigg).
\end{align}}
\end{lemma}

\begin{proof}
The characteristic function $\phi_I(t,\la,\zeta)$ of the interference term $I = \sum_{z_j \in \Psi} \frac{h_j}{d_j^{2\al}}$ is given by
\begin{align}
&\phi_I(t,\la,\zeta) = \EE_I\left[\exp\left\{\imath t I\right\}\right] = \EE_{\Psi, h_j} \left[\exp\left\{\imath t \sum_{z_j \in \Psi} \frac{h_j}{d_j^{2\al}}\right\} \right]\nonumber\\
&= \exp\left\{2\pi\la \intd \left(\EE_{h_j} \left[\exp\left\{\frac{\imath t h_j}{u^{2\al}}\right\}\right] - 1\right)u \dd u\right\}\label{pgfl}\\
&= \exp\Bigg\{\!\pi\la \Bigg(\!\zeta^2\!-\!\xi^2\!+\!\int_0^\infty\!\!\! \intd\! 2 u\exp\left\{\!\frac{\imath t h}{u^{2\al}}\!\right\}\! \dd u f_h(h,\rho)\dd h\!\Bigg)\!\Bigg\},\label{fad1}
\end{align}
where \eqref{pgfl} follows from the probability generating functional of a PPP \cite{HAE}, \eqref{fad1} follows by the integral evaluation of $u$ and by unconditioning on $h$ and $f_h(h,\rho)$ is given by \eqref{pdf_fad0} or \eqref{pdf_fad1}. Using $\int_a^b f(x)\dd x = \int_0^b f(x)\dd x-\int_0^a f(x)\dd x$, \eqref{char} can be derived by applying the transformation $v \to -\imath t \frac{h}{u^{2\al}}$ and by the definition of the upper incomplete gamma function \cite[8.350.2]{GRAD}.\vspace{-3mm}
\end{proof}

\subsection{SDMA and FDMA Scenario}
Initially, the performance with the two multiple-access schemes, SDMA and FDMA, is considered. The analysis is executed for a typical sensor, denoted by $o$.

\begin{theorem}\label{thm1}
The probability of decoding a random sensor is
\begin{align}\label{gil}
\Pi_d = \frac{1}{2} - &\frac{1}{\al\pi(\xi^2-\zeta^2)\tau^{\frac{1}{\al}}} \nonumber\\
&\quad\times \int_0^\infty \frac{1}{t}\Im\left[(\imath t)^{\frac{1}{\al}} \nu(t) \chi(t) \phi_{I_o}(t,\la',\zeta)\right] \dd t,
\end{align}
where $\tau$ is the target SINR threshold, $\nu(t) = \exp\left\{\frac{\imath t \s}{\be G P}\right\}$,
{\small\begin{align}\label{chi_fun}
\!\chi(t) \!=\! \int_0^\infty \!\!h^\frac{1}{\al} f_h(h,\rho) \Bigg(\!\Gamma\left[\!-\frac{1}{\al}, \!\frac{\imath t h}{\tau\xi^{2\al}}\right] - \Gamma\left[\!-\frac{1}{\al}, \frac{\imath t h}{\tau \zeta^{2\al}}\right]\!\Bigg) \dd h,
\end{align}}
and $\phi_{I_o}(t,\la',\zeta)$ is given by Lemma \ref{lem1} with $\la' = \frac{p_c}{D}\la$.
\end{theorem}

\begin{proof}
The decoding probability of the typical sensor is expressed as $\PP(\gamma_o > \tau)$ where $\tau$ is a predefined threshold. Thus, $\PP\left(\frac{\be G P h_o}{d_o^{2\al}(\s+ \be G P I_o)} > \tau\right) = F_{I_o}\left(\frac{h_o}{d_o^{2\al}\tau}-\frac{\s}{\be G P}\right)$ where $F_{I_o}(x)$ is the cumulative distribution function (cdf) of $I_o$. By using the Gil-Pelaez inversion theorem \cite{MDR}, we have $F_{I_o}(x) = 1/2-(1/\pi) \int_0^\infty (1/t) \Im\left[\exp\{-\imath t x\} \phi_{I_o}(t,\la,\zeta)\right] \dd t$, where $\phi_{I_o}(t,\la,\zeta)$ is the characteristic function of $I_o$ evaluated at $t$, derived in Lemma \ref{lem1}. Since a sensor interferes with probability $p_c$ and the beamwidth of the reader's sectorized antenna is $2\pi/D$, the density of the interfering sensors is thinned by $p_c/D$. Therefore, the characteristic function is given by $\phi_{I_o}(t,p_c\la/D,\zeta)$. We derive $\EE[\exp\{-\imath t x\}]$, where $x = \frac{h_o}{d_o^{2\al}\tau}-\frac{\s}{\be G P}$, as follows
\begin{align}
&\EE[\exp\{-\imath t x\}] = \exp\left\{\frac{\imath t \s}{\be G P}\right\} \EE_{\Psi,h_o}\left[\exp\left\{-\frac{\imath t h_o}{d_o^{2\al}\tau}\right\}\right]\nonumber\\
&= \frac{\exp\left\{\frac{\imath t \s}{\be G P}\right\}}{\pi(\xi^2-\zeta^2)} \!\int_0^{2\pi}\! \intd \EE_{h_o}\left[\exp\left\{-\frac{\imath t h_o}{u^{2\al}\tau}\right\}\right] u \dd u \dd \theta,\label{dist}
\end{align}
which follows by unconditioning on the distance and using \eqref{dist_pdf}. Following similar steps to the proof of Lemma \ref{lem1}, the final expression is deduced by some trivial algebraic operations.
\end{proof}

The expression of Theorem \ref{thm1} consists of the terms $\nu(t)$, $\chi(t)$ and $\phi_{I_o}(t,\la',\zeta)$, which correspond to the effect of noise/backscatter power, direct link and interference, respectively, on the decoding probability $\Pi_d$. It is clear that an increase in the transmit power $P$ will have a positive effect as it will alleviate the noise effects. For high transmit power, i.e. $P \to \infty$ we have $\nu(t) \to 1$, which provides the upper bound of $\Pi_d$ in terms of transmit power. On the other hand, for small $P$ or $\be$, $\nu(t)$ increases and, as a result, $\Pi_d$ decreases. The term $\chi(t)$ has a positive effect on the performance with a decrease of $\xi$, $\al$ and $\tau$ since these will improve the quality of the direct link. Similarly, $\phi_{I_o}(t,\la',\zeta)$ improves the performance with a decrease of $\la$ and $\xi$ which will reduce the interfering signals. These observations are validated in Section \ref{num}.

\begin{figure*}[t]
  \begin{minipage}{0.32\linewidth}\centering
    \includegraphics[width=\linewidth]{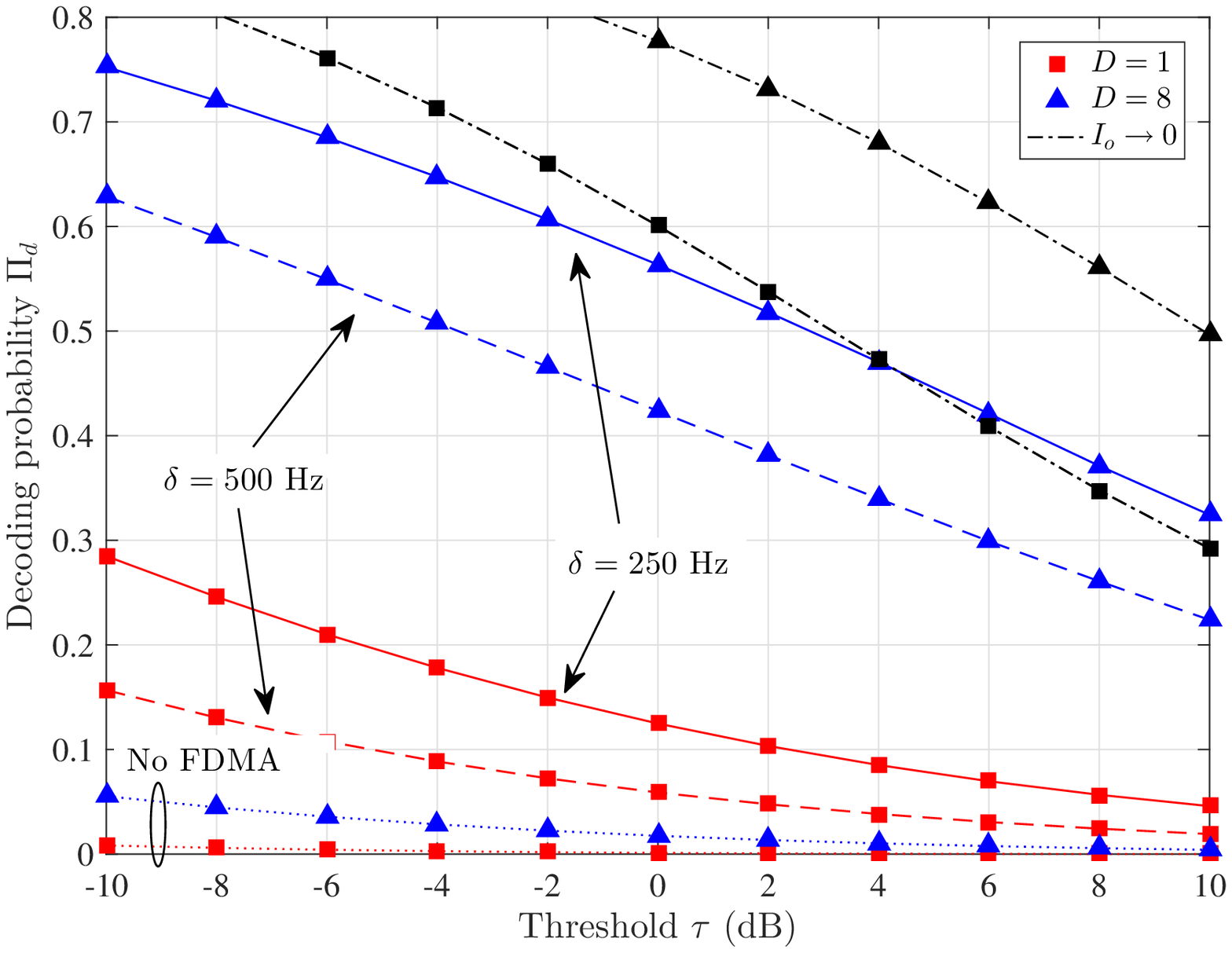}
    \captionof{figure}{Decoding probability vs threshold $\tau$; $P = 20$ dB, $\la = 1$, $\rho = 1$.}\label{fig1}
  \end{minipage}\hspace{2mm}
  \begin{minipage}{0.32\linewidth}\centering\vspace{0.7mm}
    \includegraphics[width=\linewidth,height=4.4cm]{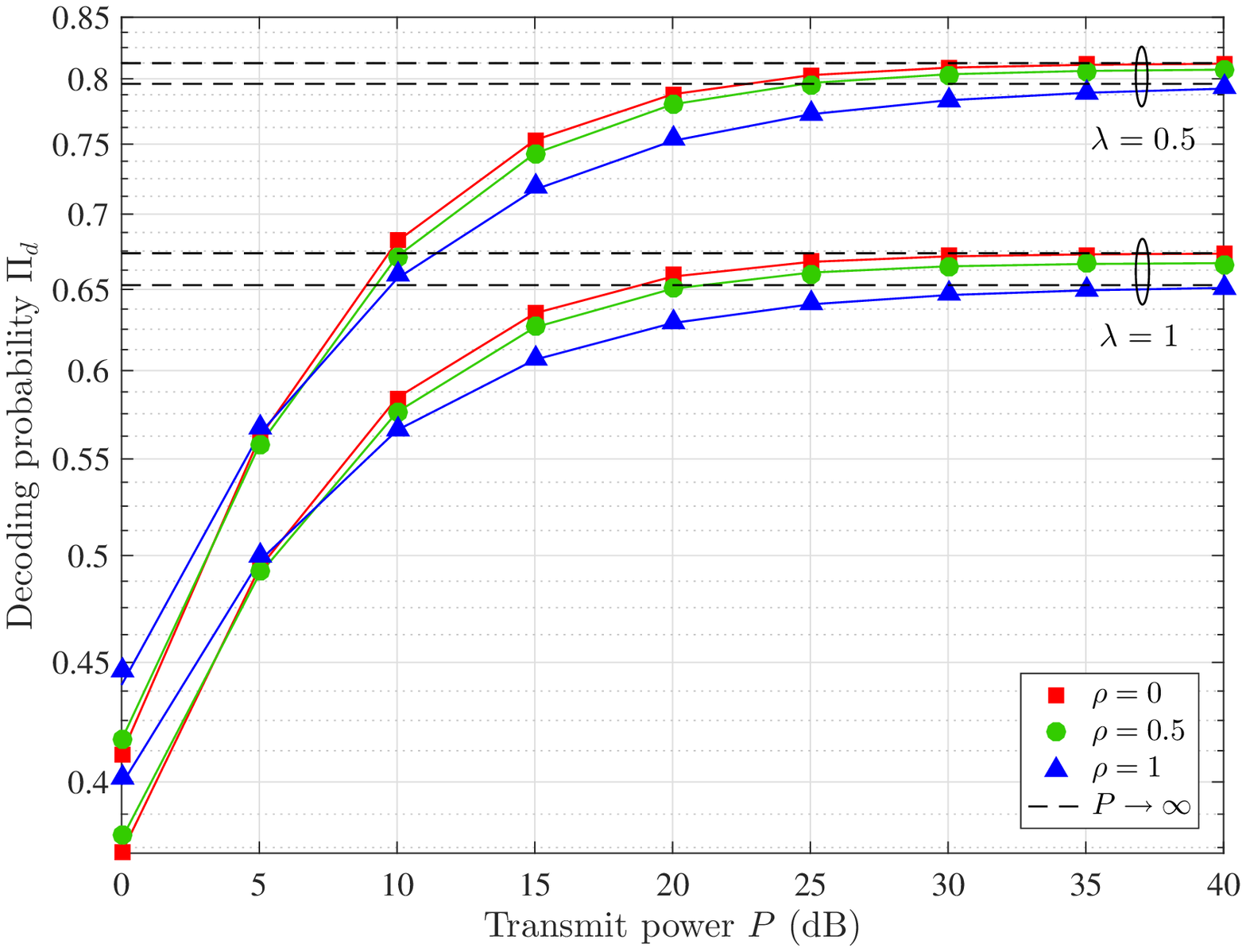}\vspace{0.3mm}
    \captionof{figure}{Decoding probability vs $P$; $\tau = -10$ dB, $D = 8$, $\delta = 500$ Hz.}\label{fig2}
  \end{minipage}\hspace{2mm}
  \begin{minipage}{0.32\linewidth}\centering\vspace{1.8mm}
    \includegraphics[width=\linewidth]{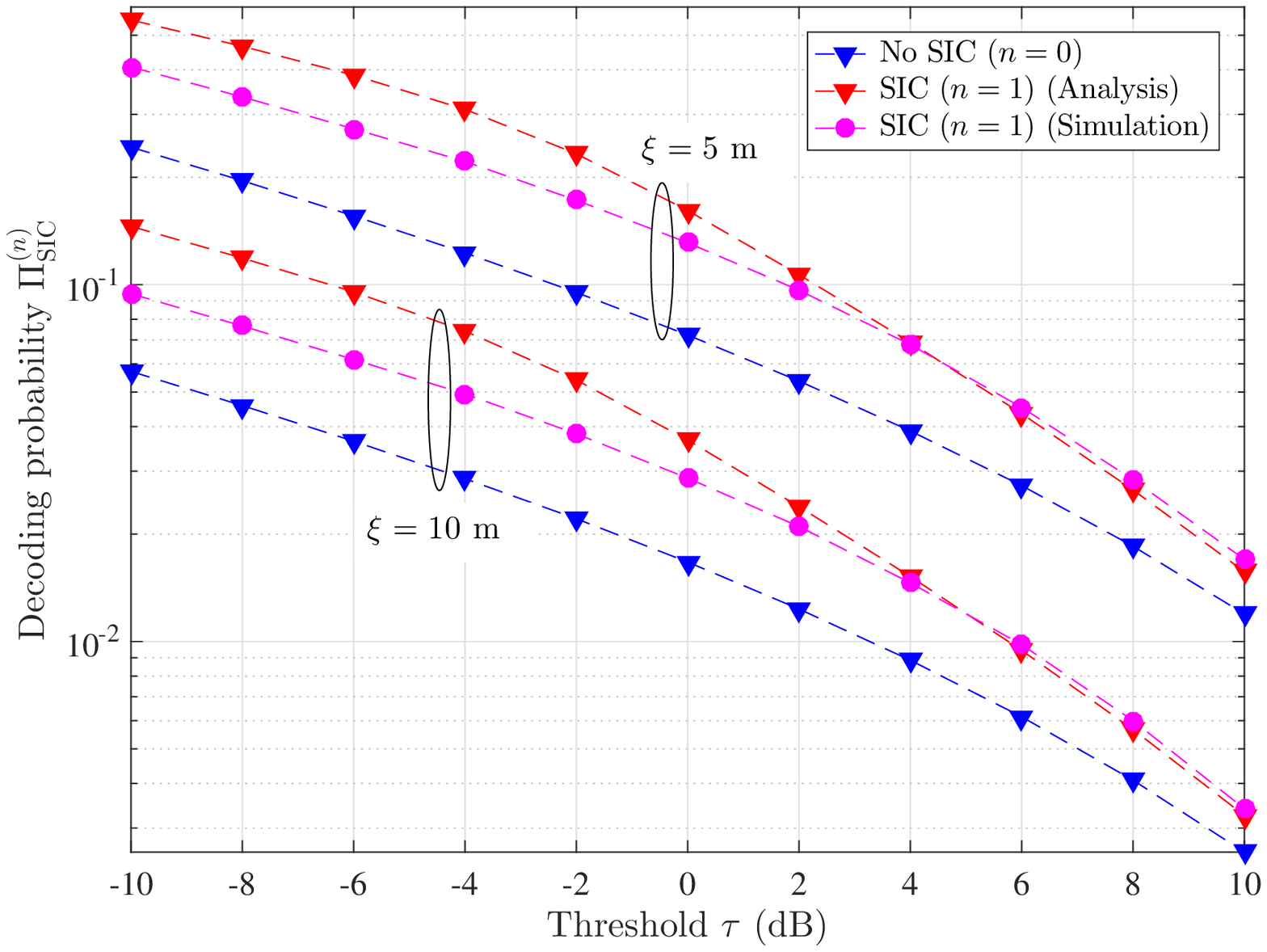}\vspace{0.5mm}
    \captionof{figure}{Decoding probability vs threshold $\tau$ with SIC; no FDMA, $D = 8$.}\label{fig3}
  \end{minipage}\vspace{-4mm}
\end{figure*}

\subsection{SIC Scenario}
We now look at the case where the reader employs SIC. Here, we consider a fading-free scenario to simplify the analysis; this is a reasonable assumption in environments, which are characterized by strong line-of-sight links. In a fading-free scenario, the order statistics are defined by the distances. In what follows, we assume $d_i \leq d_{i+1}$ $\forall z_i \in \Psi$ and use the same notation as above but drop $h$ where appropriate. We need to derive the decoding probability of the $n$-th closest sensor as well as the decoding probability of a random sensor after the $n$-th interfering signal has been removed. These can be easily deduced from Theorem \ref{thm1}. Specifically, the probability of decoding the $n$-th closest sensor $\mathsf{P_c}^{(n)}$, can be written as
\begin{align}\label{pcancel}
\!\mathsf{P_c}^{(n)} \!=\! \frac{1}{2} \!-\! \frac{1}{2\al\pi}\! \int_0^\infty \!\!\Im\left[\left(\frac{\imath t}{\tau}\right)^{\frac{1}{\al}}\! \exp\left\{\frac{\imath t \s}{\be G P}\right\} \omega(t)\right]\! \frac{\dd t}{t},\!
\end{align}
with $\omega(t) = \intd \exp\left\{-\frac{\imath t}{r^{2\al}\tau}\right\} \phi_{I_n}(t,\la',r) f_r(r) \dd r$, where $\phi_{I_n}(t,\la',r)$ is given by Lemma \ref{lem1} and $f_r(r)$ is the pdf of the distance $r$ from the origin to the $n$-th nearest sensor \cite{HAE}
\begin{align}\label{pdf2}
f_r(r) = \frac{2(\pi\la)^n}{(n-1)!} r^{2n-1} \exp\{-\pi\la (r^2-\zeta^2)\}.
\end{align}
Note that \eqref{pdf2} is slightly different to the one in \cite{HAE}, as we do not consider distances less than $\zeta$. Similarly, the decoding probability of a random sensor after the $n$-th interfering signal has been removed, denoted by $\mathsf{P_d}^{(n)}$, is given by \eqref{pcancel} but with $\omega(t) = \chi(t) \intd \phi_{I_n}(t,\la',r) f_r(r) \dd r$, where $\chi(t)$ given by \eqref{chi_fun}.

\begin{proposition}\label{prop1}
The probability of decoding a random sensor after attempting to cancel up to $n$ interfering signals is
\begin{align}
\Pi_{\rm SIC}^{(n)} \!=\! \Pi_d + \sum_{i=1}^n \!\left(\prod_{j=0}^{i-1} \left(1\!-\!\mathsf{P_d}^{\!(j)}\right)\right)\! \left(\prod_{j=1}^i \mathsf{P_c}^{\!(j)}\right) \!\mathsf{P_d}^{(i)},
\end{align}
where $\mathsf{P_d}^{(0)} = \Pi_d$ is given by Theorem \ref{thm1} (without fading).
\end{proposition}
The above expression is derived by assuming that the events described in Section \ref{techniques} are independent; this approximation is tight for high thresholds \cite{CP}.

Finally, we consider the asymptotic scenario $I_o \to 0$, i.e. a noise-limited scenario. This refers to scenarios with very small collision probability, i.e. $p_c \to 0$ or scenarios with a massive number of sectorized antennas, i.e. $D \to \infty$. This could also be achieved by employing SIC to decode and cancel all interfering signals; however, probabilistically this scenario is not always possible. For the asymptotic case, we have the following proposition for the fully correlated case, i.e. $\rho = 1$.

\begin{proposition}
In a noise-limited scenario, the decoding probability of a random sensor when $\rho = 1$ is\vspace{-2mm}

{\small\begin{align}
\Pi_d^{I_o \to 0} = \frac{2\left(\frac{\be G P}{\tau\s}\right)^{\frac{1}{\al}}}{\al(\xi^2 -\zeta^2)} \Bigg(\!\Gamma&\left[\!\frac{2}{\al}, \zeta^\al \sqrt{\frac{\tau \s}{\be G P}}\right] \!-\! \Gamma\left[\!\frac{2}{\al}, \xi^\al \sqrt{\frac{\tau \s}{\be G P}}\right]\!\Bigg)\label{asym}.
\end{align}}
\end{proposition}

\begin{proof}
For $I_o \to 0$, the decoding probability is given by
\begin{align}
&\PP\left(\frac{\be G P h_o}{d_o^{2\al}\s} > \tau\right) = \PP\left(h_o > \frac{\tau d_o^{2\al}\s}{\be G P}\right)\nonumber\\
&= \int_0^{2\pi} \intd \exp\left\{-d^\al \sqrt{\frac{\tau\s}{\be G P}}\right\} d f_d(d) \dd d,
\end{align}
which follows from the cdf $F_h(h,1) = \int_0^x f_h(h,1) \dd h = 1 - \exp\{-\sqrt{x}\}$ and $f_d(d)$ is given by \eqref{dist_pdf}. Finally, \eqref{asym} follows from \cite[3.381.8]{GRAD}.
\end{proof}

Note that for high transmit power, i.e. $P \to \infty$, we have $\Pi_d^{I_o \to 0} \to 1$ by using $\lim_{x \to 0} \Gamma[a,x] \to \Gamma[a]-\frac{x^a}{a}$. The expression for $0 \leq \rho < 1$ can be derived in a similar way but is omitted due to space restrictions.\vspace{-1mm}

\section{Numerical Results}\label{num}
We validate and evaluate our proposed model with computer simulations. Unless otherwise stated, the following parameters are used: $\lambda = 1$, $P = 20$ dB, $\zeta = 1$ m, $\xi = 10$ m, $\alpha = 2.5$, $\sigma^2 = -30$ dB, $BW = 12$ kHz, $\epsilon = 0.1$, and $\beta = 0.6$. The lines (dashed or solid) and markers in the plots represent the analytical and simulation results, respectively.

Fig. \ref{fig1} illustrates the decoding probability with respect to the threshold $\tau$ for different values of $\delta$ and $D$ and $\rho = 1$. It is clear that the performance increases with the employment of the FDMA scheme. Specifically, while $\delta$ becomes narrower the decoding probability increases, which is expected since the narrower $\delta$ is, the smaller the probability of collision becomes. Similarly, the performance increases with  the number of sectorized antennas $D$, as a smaller sector implies fewer collisions from other sensors. It is important to point out how critical $\delta$ and $D$ are to the system's performance since, for $D = 1$ and no FDMA, the decoding probability is close to zero. The scenario $I_o \to 0$ provides high performances as expected. However, for high threshold values, the case $D = 1$, $I_o \to 0$ is outperformed by $D = 8$, $\delta = 250$ Hz due to higher antenna gain when $D = 8$. Finally, our theoretical results (lines) perfectly match the simulation results (markers) which validates our analysis. Fig. \ref{fig2} depicts the effect of the correlation parameter $\rho$ and the transmit power $P$ on the decoding probability. It is clear that the correlated case performs better for small values of $P$ whereas it provides the lowest performance for higher values. As expected, the performance improves with $P$ and as $P$ increases, it converges to the upper bound. Also, the performance increases for a smaller density for all values of $\rho$. Finally, Fig. \ref{fig3} shows the decoding probability when the reader employs SIC. We consider a low-complexity scenario with $n=1$. The assumption in Proposition \ref{prop1} is validated since the approximation is tight for high threshold values. It is obvious from the figure that the employment of SIC provides significant gains to the performance even when just the signal from the first nearest sensor is cancelled.

\section{Conclusion}
In this letter, we studied backscatter communication sensor networks with spatial randomness. We investigated three techniques for collision resolution: ultra-narrow band transmissions by the sensors as well as antenna directionality and SIC at the reader. Mathematical expressions for the decoding probability were derived using tools from stochastic geometry. Our results showed that a combination of these techniques is needed to achieve significant gains.

\end{document}